\documentclass{amcOF}
\usepackage{amsmath}
\usepackage{paralist}
\usepackage{cite}
\usepackage{amscd,amssymb,latexsym}

\usepackage[colorlinks=true]{hyperref}
\hypersetup{urlcolor=blue, citecolor=red}

\def\currentvolume{X}
 \def\currentissue{X}
  \def\currentyear{20XX}
   
    \def\ppages{X--XX}
     \def\DOI{???/amc.???}

\newtheorem{theorem}{Theorem}[section]
\newtheorem{corollary}{Corollary}

\newtheorem{lemma}[theorem]{Lemma}
\newtheorem{proposition}{Proposition}

\theoremstyle{definition}

\newtheorem*{remark}{Remark}
\newtheorem{example}{Example}

\newcommand{\ord}{\operatorname{ord}}
\newcommand{\Tr}{\operatorname{Tr}}

\allowdisplaybreaks

\title[Ternary Primitive LCD BCH  codes]{Ternary Primitive LCD BCH  codes }

\author[Xinmei Huang, Qin Yue, Yansheng Wu and Xiaoping Shi]{}


\thanks{This paper was supported by the National Natural Science Foundation of China (No.61772015), the Foundation of Science and Technology on Information Assurance Laboratory (No.KJ-17-010) and the Foundation of Jinling Institute of Technology (No.JIT-B-202016, No.JIT-FHXM-2020). Y. Wu was sponsored by NUPTSF (No. NY220137). (Corresponding author:Yansheng Wu.) }

\begin{document}
\maketitle

\centerline{\scshape Xinmei Huang}
\medskip
{\footnotesize
  \centerline{Department of Mathematics, Jinling Institute of Technology}
   \centerline{Nanjing, 211169, P. R. China}
   \centerline{State Key Laboratory of Cryptology, P. O. Box 5159}\centerline{ Beijing, 100878, China}
}

\medskip

\centerline{\scshape Qin Yue}
\medskip
{\footnotesize
  \centerline{Department of Mathematics, Nanjing University of Aeronautics and Astronautics}
   \centerline{Nanjing, Jiangsu, 211100, China}
  }

\medskip

\centerline{\scshape Yansheng Wu}
\medskip
{\footnotesize
 \centerline{School of Computer Science, Nanjing University of Posts and Telecommunications}
  \centerline{Nanjing 210023, P. R. China}
   }

\medskip

\centerline{\scshape Xiaoping Shi}
\medskip
{\footnotesize
 \centerline{Department of Mathematics, Nanjing Forestry University}
  \centerline{Nanjing 210037, P. R. China}
   }
\bigskip
\centerline{(Communicated by )}

\begin{abstract}
Absolute coset leaders were first proposed by the authors which  have advantages in constructing binary LCD BCH  codes.
As a continue work, in this paper we focus on ternary linear codes.
Firstly, we find the largest, second largest, and third largest absolute coset leaders of ternary primitive BCH codes.
Secondly, we  present three classes of ternary primitive BCH codes and determine their weight distributions.
Finally, we obtain some  LCD BCH codes and  calculate some weight distributions.
However,  the calculation of weight distributions  of two of these codes is equivalent to that of  Kloosterman sums.
\end{abstract}

\section{Introduction}

Let $\Bbb F_q$ be a finite field with $q$ elements, where $q$ is a prime power. An $[n,k,d]$ linear code $\mathcal{C}$  over $\Bbb F_q$ is a linear subspace of $\Bbb F_q^n$ with dimension $k$ and minimum (Hamming) distance $d$. Let $A_i$ denote the number of codewords in $\mathcal C$
 with Hamming weight $i$. The weight enumerator of
 $\mathcal C$ is defined by
 $1+A_1z+A_2z^2+\cdots+A_nz^n.$
The sequence $(1, A_1, A_2, \ldots, A_n)$ is called the weight distribution of
 $\mathcal C$. A code $\mathcal{C}$ is $t$-weight if the number of nonzero $A_{i}$ in the sequence $(A_1, A_2, \ldots, A_n)$ is equal to $t$.

We define the standard Euclidean inner product of the $\Bbb F_q$-vector space $\Bbb F_q^n$ as follows:
for ${\bf a}=(a_0,\ldots, a_{n-1}), {\bf c}=(c_0,\ldots, c_{n-1})$, $\langle {\bf a}, {\bf c}\rangle={\bf a}{\bf c}^T=\sum_{i=0}^{n-1}a_ic_i$.
Let $\mathcal C$ be an $[n,k]$ linear code, its dual code is defined as follows: $$\mathcal{C}^{\perp}=\{{\bf a}\in \Bbb F_q^n: {\bf a}{\bf c}^T=0 \mbox{ for all }{\bf c}\in \mathcal C\}.$$
If the code $\mathcal{C}$ satisfies the condition that each  codeword $(c_0,c_1,\ldots,c_{n-1})\in  \mathcal{C}$
implies $(c_{n-1},c_0,c_1,\ldots,c_{n-2})\in \mathcal{C}$, then $ \mathcal{C}$ is said to be a cyclic code. A cyclic code $\mathcal{C}$ of length $n$ over $\Bbb F_q$  corresponds to an ideal of the quotient ring $\Bbb F_q[x] /\langle x^n-1\rangle$. Furthermore, $\Bbb F_q[x] /\langle x^n-1\rangle $ is a principle ideal ring, and $\mathcal{C}$ is generated by a monic divisor $g(x)$ of ${x^n} - 1$. In this situation, $g(x)$ is called the generator polynomial of the code $\mathcal{C}$ and we write $\mathcal{C} = \langle g(x)\rangle$.

Let $\Bbb Z_n=\{0,1,\ldots,n-1\}$ be the ring of integers modulo $n$. For $s\in \Bbb Z_n$,
 assume that  $l_s$ is the smallest positive integer such that
$q^{l_s}s\equiv s\pmod n$. Then  the $q$-cyclotomic coset of $s$ modulo $n$
is defined by $$C_s=\{s, sq, \cdots, sq^{l_s-1}\}_{\mod n}\subset \Bbb Z_n $$ and $ |C_s|=l_s$.  The  smallest  integer in $C_s$  is called the {\bf  coset leader} of $C_s$ (see \cite{LDL2}). In the paper \cite{HYWS}, the authors gave  a new definition to investigate  LCD BCH codes. Define that the smallest  integer  in the set $\{k,n-k: k\in C_s\}$ is called the {\bf absolute coset leader} of $C_s$.

Let  $m=\ord_n(q)$ be the multiplicative order of $q$ modulo $n$ and  $\gamma $  a primitive element of $\Bbb F_{q^m}$.  Then  $\alpha=\gamma^{\frac{q^m-1}n}$ is of order $n$.
A cyclic code  $\mathcal{C}_{(q,n,\delta,b)}=\langle g(x)\rangle$ of length $n$ over $\Bbb F_q$ is called a BCH code with the designed distance $\delta$ if its generator polynomial is of the form $$g(x)=\prod\limits_{i \in Z} {(x - {\alpha ^i}}),\ Z=C_{b+1}\cup C_{b+2}\cup\cdots \cup C_{b+\delta-1},$$ where $Z$ is called the defining set of $\mathcal C_{(q,n,\delta,b)}$. If $n=q^m-1$, we call $\mathcal{C}_{(q,n,\delta,b)}$ a primitive BCH code. If $b=0$, $\mathcal{C}_{(q,n,\delta,b)}$ is called a narrow-sense BCH code; otherwise,  it is called a non-narrow-sense BCH code. The dimension of $\mathcal C_{(q,n,\delta,b)}$ is
$\dim(\mathcal C_{(q,n,\delta,b)})= n-|\bigcup_{i=b+1}^{b+\delta-1}C_i|$.
Thus, to determine the dimension of the code $\mathcal C_{(q,n,\delta,b)}$, we only need to find out all coset leaders and
  cardinalities of the $q$-cyclotomic cosets.

LCD cyclic codes named  reversible codes were first studied by Massey for data
storage applications \cite{M1}. An application
of LCD codes against side-channel attacks was  investigated by Carlet and Guilley, and
several constructions of LCD codes were presented in \cite{CG}. Several constructions
of LCD MDS codes were presented in \cite{CMTQ,CL,J,L, S}. Tzeng and Hartmann proved
that the minimum distance of a class of LCD cyclic codes is
greater than the BCH bound \cite{TH}. Several investigations of LCD BCH codes were studied in \cite{HYWS,LDL2,LLDL, WY,YLLY}.
Parameters and the weight distributions of BCH codes are studied in  \cite{D,DFZ,LDXG,LDL1,LYL, LWL, LYW}. LCD codes in Hermitian case were studied in \cite{CMTQ,L}. In \cite{CMTQP}, Carlet et al. completly determined all $q$-ary($q>3$) Euclidean LCD codes and all $q^2$-ary ($q>2$) Hermitian LCD codes for all parameter.  Some binary and ternary LCD codes were investigated in \cite{HYWS,ZLTD}. In \cite{HYWS}, the authors proposed a new conception, named absolute coset leader, and constructed some binary LCD BCH codes. In this paper, we shall investigate the ternary case.

The remainder of the paper is organized as follows. In Section 2, some fundamental definitions and results are introduced. In Section 3, the  largest, second largest, and third largest absolute coset leaders are presented for ternary primitive BCH codes.  In Section 4, some BCH codes and their weight distributions are presented. Also, LCD BCH codes are constructed and their parameters are determined,  some weight distributions are calculated, the determination of the others is equivalent to the computing of  Kloosterman sums. In Section 5, we conclude this paper.

\section{Preliminaries}

A linear code $\mathcal{C}$  over $\Bbb F_q$ is called a linear code with complementary
dual code (LCD for short)  if $\mathcal{C}\bigcap \mathcal{C}^{\perp}= \{0\}$, where $\mathcal{C}^{\perp}$ denotes the Euclidean dual of $\mathcal{C}$.

Let $f(x) = x^t +a_{t-1}x^{t-1}+\cdots+a_1x_1+a_0 $ be a monic  polynomial over $\Bbb F_q$ with $a_0 \neq 0$. The reciprocal polynomial  of $f(x)$ is defined by
$\widehat{f}(x)=a_0^{-1}x^tf(x^{-1})$.
Then we have the following lemma that characterizes LCD cyclic codes over  $\Bbb F_q$.

\begin{lemma}\cite{YM} {\rm
Let $\mathcal{C}$ be a cyclic code of length $n$ over $\Bbb F_q$ with generator polynomial $g(x)$ and $\gcd(n,q)=1$.
Then the following statements are equivalent.
\begin{enumerate}
\item $\mathcal{C}$ is an LCD code.
\item $g(x)$ is self-reciprocal, i.e., $g(x) = \widehat{g}(x)$.
\item $\alpha^{-1}$ is a root of $g(x)$ for every root $\alpha$ of $g(x)$.
\end{enumerate}}
\end{lemma}

Let $\Bbb F_q$ be the finite field with $q$ elements, where $q$ is  a power of  a prime number $p$. The canonical additive character of $\Bbb F_q$ is defined as follows:
$$\chi:\Bbb F_q\rightarrow \mathbb{C}^*, \chi(x)=\zeta_p^{\Tr_{q/p}(x)},$$
where $\zeta_p=e^{\frac{2\pi i}{p}}$ is a $p$-th primitive root of unity and $\Tr_{q/p}$ denotes the trace function from $\Bbb F_q$ to $\Bbb F_p$. The orthogonal property of additive characters which can be found in  \cite{LN}
\begin{equation*}
\underset{x\in \Bbb F_q}\sum\chi(ax)=\left\{
\begin{aligned}
q& ~\text{ if }~a=0,\\
0& ~\text{ if }~a\in \Bbb F_q^*.\\
\end{aligned}
\right.
\end{equation*}
Let $\psi:\Bbb F_q\rightarrow\Bbb C^*$ be a multiplicative character of $\Bbb F_q^*$. The trivial multiplicative character
$\psi_0$ is defined by $\psi_0(x) = 1$ for all $x\in \Bbb F_q^*$ . For two multiplicative characters $\psi$,$\psi'$ of $\Bbb F_q^*$ ,
we define the multiplication by setting $\psi\psi'(x) = \psi(x)\psi'(x)$ for all $x \in \Bbb F_q^*$. Let $\bar{\psi}$ be
the conjugate character of $\psi$ defined by $\bar{\psi}(x) = \overline{\psi(x)}$, where $\overline{\psi(x)}$ denotes the complex conjugate of $\psi(x)$. It is easy to deduce that $\psi^{-1}=\bar{\psi}$. It is known \cite{LN} that all the
multiplicative characters form a multiplication group $\hat{\Bbb F_q^*}$ which is isomorphic to $\Bbb F_q^*$. The
 orthogonal property of multiplicative characters \cite{LN} is:
\begin{equation*}
\underset{x\in \Bbb F_q^*}\sum\psi(x)=\left\{
\begin{aligned}
q-1& ~\text{ if }~ \psi=\psi_0, \\
0& ~\text{ otherwise}. \\
\end{aligned}
\right.
\end{equation*}
The Gauss sum over $\Bbb F_q$ is defined by
\begin{equation*}
    G(\psi,\chi)=\sum_{x\in \Bbb F_q^*}\chi(x)\psi(x).
\end{equation*}
It is easy to see that $G(\psi_0,\chi) = -1 $ and $G(\bar{\psi},\chi) = \psi(-1)\overline{G(\psi,\chi)}$. Gauss sums can be viewed as
the Fourier coefficients in the Fourier expansion of the restriction of $\psi$ to $\Bbb F_q^*$ in terms of
the multiplicative characters of $\Bbb F_q$, i.e., for $x\in \Bbb F_q^*$,
\begin{equation}\label{f1}
   \chi(x)=\frac1{q-1} \sum_{x\in \hat{\Bbb F_q^*}}G(\bar{\psi},\chi)\psi(x).
\end{equation}

Using (\ref{f1}), we can get the following results.

\begin{lemma}\cite{LN} {\rm
Let $\chi$ be a nontrivial additive character of $\Bbb F_q$, $n\in \Bbb N$, and $\lambda$ a multiplicative character of $\Bbb F_q$ of order $d=gcd(n,q-1)$. Then
\begin{equation*}
    \sum_{x\in \Bbb F_q}\chi(ax^n+b)=\chi(b)\sum_{j=1}^{d-1}\bar{\lambda}^j(a)G(\lambda^j,\chi)
\end{equation*}
for any $a,b\in\Bbb F_q$ with $a\neq 0$. }
\end{lemma}

In general, the explicit determination of Gauss sums is a difficult problem.  For future use, we state the quadratic Gauss sums here.

\begin{lemma}\cite{LN} {\rm
 Let $\Bbb F_q$ be a finite field with $q=p^s$, where $p$ is an odd prime and $s\in \Bbb N$. Let $\eta$ be the quadratic character of $\Bbb F_q$ and let $\chi$ be the canonical additive character of $\Bbb F_q$. Then
\begin{equation*}
    G(\eta,\chi)=
    \left\{
    \begin{array}{ll}
      (-1)^sq^{1/2} &\mbox{ if } p \equiv1\pmod 4,    \\
      (-1)^{s-1}(\sqrt{-1})^sq^{1/2} &\mbox{ if }p \equiv3\pmod 4.
    \end{array}
    \right.
\end{equation*}}
\end{lemma}

\section{Absolute coset leaders of ternary BCH codes}

In this section, we will find the  first, second and third largest absolute coset leaders of ternary cyclic BCH codes of length $n=3^m-1$ over $\Bbb F_3$, where $m=\ord_n(3)$.

In \cite{LDXG,LDL1,YLLY}, the authors determined the largest and second largest  coset leaders of BCH codes in three cases:
 (1) $n=q^m-1$; (2) $n=\frac{q^m-1}{q-1}$; (3)  $n = q^l + 1$.
In \cite{HYWS}, the authors determined the largest and second largest  absolute coset leaders of binary BCH codes.

Before presenting our results,  we describe some notations.
The $3$-adic expansion of an integer $i\in \Bbb Z_n$ is denoted by  
\begin{equation*}
  i=i_0+i_13+\cdots+i_{m-1}3^{m-1}\triangleq(i_0,i_1,\ldots,i_{m-1}),
\end{equation*}
where each $0\le i_t \le 2$.

According to the definition of  absolute coset leaders, we can get the following proposition.

\begin{proposition}\cite{HYWS} {\rm Let the {\bf absolute coset leader} of $C_\delta$ be $\delta$ and $n=q^m-1$.

(1) Then $\delta\leq \frac n 2$.

(2) If $n-\delta \notin C_\delta$, then $C_{n-\delta}\neq C_\delta$, $|C_{n-\delta}|=|C_\delta|$, and $C_{n-\delta}$ has the same absolute coset leader $\delta$ as one in  $C_\delta$. }
\end{proposition}

\begin{theorem} \label{th1} {\rm Let $q=3$,   $m$  a positive integer, and   $n=q^m-1$.
 Then  $\delta_1=\frac{3^m-1}{2}$ is  the  {\bf largest  absolute coset leader} among all  $3$-cyclotomic cosets, $C_{\delta_1}=\{\delta_1\}$, and $|C_{\delta_1}|=1$.}
\end{theorem}

\begin{proof}
 We shall  verify that  $\delta_1$ is the largest absolute coset leader among all $3$-cyclotomic cosets $C_s, 0\le s\le n-1$.

There are two  $3$-adic expansions of $n$ and $\delta_1$:
\begin{eqnarray}n&=&(2, 2, 2, 2, \ldots, 2, 2, 2), \\
 \delta_1&=&(1, 1, 1, 1, \ldots, 1, 1 ). \nonumber  \end{eqnarray}

Firstly, we prove that $\delta_1$ is the absolute  coset leader of the $q$-cyclotomic cosets $C_{\delta_1}$.
For $ 1 \leq l \leq m-1$,
$$3^l\delta_1 \pmod n\equiv(1, 1, 1, 1, \ldots, 1, 1 )$$
Hence $C_{\delta_1}=\{\delta_1\}$ only has one element, i.e. $|C_{\delta_1}|=1$.

Secondly, we will show  that $\delta_1$ is  the largest absolute coset leader among all $q$-cyclotomic cosets.

For $0\le i\le n-1$, there is a $3$-adic expansion:
$$i=i_0+i_13+\cdots+i_{m-1}3^{m-1}=(i_0,i_1,\ldots, i_{m-1}),$$
where each $i_t\in \{0,1,2\}$, $t=0,1,\ldots, m-1$.

If the expansion of $i$ has $0$. Without loss of generality, let  $i=(\ldots, 0,\ldots)$. Then there is an integer $l$, $0\le l\le m-1$, such that $3^li \pmod n \equiv(\ldots,0)\in C_i$, so $3^li\pmod n <\delta_1$ by (3.1). Hence the absolute coset leader in $C_i$ is less than $\delta_1$.

If the expansion of $i$ has $2$. Similarly, let  $i=(\ldots, 2,\ldots)$, there is an integer $l$, $0\le l\le m-1$,  such that $n-3^li\pmod n <\delta_1$.
 Hence the absolute coset leader in $C_i$ is less than $\delta_1$.

Therefore, $\delta_1$ is the largest absolute coset leader among all cosets.

This completes the proof.
\end{proof}

\begin{theorem} \label{th2} {\rm Let $q=3$,   $m$  a positive integer, and   $n=q^m-1$.

(1) If $m\ge 3$ is an odd integer, then $\delta_2=\frac{3^{m-1}-1}{4}+3^{m-2}$ is the  {\bf second largest absolute coset leader},  $C_{\delta_2}\ne C_{n-\delta_2}$,  and  $|C_{\delta_2}|=|C_{n-\delta_2}|=m$.

(2) If $m\ge 2$ is an even integer, then  $\delta_2=\frac{3^m-1}{4}$  is the  {\bf second largest absolute coset leader}, $C_{\delta_2}=\{\delta_2,n-\delta_2\}$, and  $|C_{\delta_2}|=2$.}
\end{theorem}

\begin{proof}
(1) If $m$ is odd and the $3$-adic expansion of $\delta_2$ is as follows:
 $$\delta_2=\frac{3^{m-1}-1}{4}+3^{m-2}=(\underbrace{2,0, 2,0, \ldots, 2,0 }_{(m-3)/2~ ~(2,0)'s},2,1, 0),$$
then $\delta_2<\delta_1$.

Firstly, we prove that $\delta_2$ is the absolute  coset leader of the $q$-cyclotomic cosets   $C_{\delta_2}$ and $C_{n-\delta_2}$.
For $  1 \leq l \leq m-1$, if $l$ is odd, then
$$3^l\delta_2 \pmod n\equiv(\underbrace{2,0,\ldots,2,0,}_{(l-3)/2
\ (2,0)'s}2,1,0,\underbrace{2,0,\ldots,2,0}_{(m-l)/2~ ~(2,0)'s});$$
if $l$ is even, then
$$3^l\delta_1 \pmod n\equiv(0,\underbrace{2,0,\ldots,2,0,}_{(l-4)/2~ ~(2,0)'s}2,1,0,\underbrace{2,0,\ldots,2,0,}_{(m-l-1)/2~ ~(2,0)'s}2).$$
Hence $C_{\delta_2}$ has $m$ distinct elements, i.e. $|C_{\delta_2}|=m$, and  $\delta_2=\min\{k,n-k: k\in C_{\delta_2}\}$, which is the absolute coset leader in $C_{\delta_2}$.
Similarly, we can prove that $|C_{n-\delta_2}|=m$, $C_{\delta_2}\ne C_{n-\delta_2}$, and $C_{n-\delta_2}$ has also the absolute coset leader $\delta_2$.

Secondly, we prove that $\delta_2$ is the second largest absolute coset leader.

For $0\le i\le n-1$, there is a $3$-adic expansion:
$$i=i_0+i_13+\ldots+i_{m-1}3^{m-1}=(i_0,i_1,\ldots, i_{m-1}),$$
which has at least two elements among $0,1,2$. Otherwise, the expansion of $i$ has only one elements of $0,1,2$, then $i=(0,\ldots,0)<\delta_2$,  $i=\delta_1$, or $i=n-\delta_1$.

 If the expansion of $i$ has  a  consecutive form:  $(00)$, i.e.,  $i=(\ldots, 0, 0,\ldots)$.  Then there is an integer $l$, $0\le l\le m-1$, such that $3^li\pmod n \equiv(\ldots, 0,0)\in C_i$, so $3^li\pmod n <\delta_2$. Hence the absolute coset leader of $C_i$ is less than $\delta_2$. Similarly, we can prove it if the expansion of $i$ has consecutive $(22)$.

 If the expansion of $i$ has a  form:  $(110)$, i.e.,  $i=(\ldots, 1,1,0,\ldots)$.  Then there is an integer $l$, $0\le l\le m-1$, such that $3^li\pmod n \equiv(\ldots, 1,1,0)\in C_i$, so $3^li\pmod n<\delta_2$. Hence the absolute coset leader of $C_i$ is less than $\delta_2$. Similarly, we can prove if  the expansion of $i$ has a  from: $(112)$. Then there is an integer $l$, $0\le l\le m-1$, such that $3^li\pmod n \equiv(\ldots, 1,1,2)\in C_i$, so $n-3^li\pmod n<\delta_2$. Hence the absolute coset leader of $C_i$ is less than $\delta_2$.

If the expansion of $i$ has a  form: $(010)$ (or $(212)$), then there is an integer $l$ such that $3^li\pmod n<\delta_2$ (or $n-3^li\pmod n<\delta_2$, respectively). Hence  the absolute coset leader of $C_i$ is less than $\delta_2$.

If  the expansion of $i$ has not any  forms: $(00)$, $(11)$, $(22)$, $(010)$, and $(212)$. We shall  prove that the absolute coset leader of $C_i$ is less than $\delta_2$.
From the above, the expansion of $i$ is equivalent to insert some 1's into the sequence $(2,0,\ldots,2,0)$ (or $(0,2,\ldots,0,2)$).
Since $m$ is an odd integer, the number of 1's in the expansion of $i$ is an odd integer  $k$.

If $k=1$, i.e., the expansion of $i$ has only one form:  $(210)$ (or $(021)$), then there is an integer $l$, $0\le l\le m-1$, such that $3^li\pmod n \equiv \delta_2$ (or $n-3^li\pmod n \equiv \delta_2$, respectively).

If $k\ge 3$, without loss of generality,  there are two adjacent   $(210)'s$ in the expansion of $i$, i.e., $$i=(\ldots,\underbrace{2,1,0},2,0,\ldots,2,0,\underbrace{2,1,0},\ldots).$$    Then there is an integer $l$, $0\le l\le m-1$, such that
$$3^li\pmod n \equiv (\ldots,\underbrace{2,1,0},2,0,\ldots,2,0,\underbrace{2,1,0})<\delta_2.$$

Similarly, if there are two adjacent   $(012)'s$ in the expansion of $i$, i.e.,
$$i=(\ldots,\underbrace{0,1,2},0,2,\ldots,0,2,\underbrace{0,1,2},\ldots).$$    Then there is an integer $l$, $0\le l\le m-1$, such that
$$n-3^li\pmod n \equiv (\ldots,\underbrace{2,1,0},2,0,\ldots,2,0,\underbrace{2,1,0})<\delta_2.$$

Therefore $\delta_2$ is the second largest absolute coset leader for $m$ is odd.

(2) If $m$ is even, and  the expansion of $\delta_2$ is as follows:
 $$\delta_2=\frac{3^m-1}{4}=(\underbrace{2,0, 2,0, \ldots, 2,0 }_{m/2~ ~(2,0)'s}),$$
then $\delta_2<\delta_1$.

Firstly, we prove that $\delta_2$ is the absolute  coset leader of the $q$-cyclotomic cosets $C_{\delta_2}$.
For $  1 \leq l \leq m-1$, if $l$ is odd, then
$3^l\delta_2 \pmod n\equiv\delta_2$, if $l$ is even, then
$n-3^l\delta_1 \pmod n\equiv\delta_2$.
Hence $C_{\delta_2}=\{\delta_2,n-\delta_2\}$ and $|C_{\delta_2}|=2$. It is obvious that $\delta_2$ is the absolute coset leader in $C_{\delta_2}$.

Secondly,  we prove that $\delta_2$ is the second largest absolute coset leader.

For $1\le i\le n-1$, the $3$-adic expansion of $i$ is as follows:
$i=(i_0,i_1,\ldots, i_{m-1})$, which has at least two elements among $0,1,2$.

If the expansion of $i$ has a form: $(10)$ (or $(12)$). Then there is an integer $l$, $0\le l\le m-1$, such that $3^li\pmod n \equiv(\ldots, 1,0)\in C_i$ (or $3^li\pmod n \equiv(\ldots, 1,2)\in C_i$), so $3^li\pmod n <\delta_2$ (or $n-3^li\pmod n <\delta_2$, respectively). Hence,  the absolute coset leader in $C_i$ is less than $\delta_2$.

If the expansion of $i$ has a consecutive form:  $(11)$. Then the expansion of  $i$ has $(110)$ or $(112)$. From the above, the absolute coset leader in $C_i$ is less than $\delta_2$.

If the expansion of $i$ has a consecutive form:  $(00)$ (or $(22)$). Then there is an integer $l$, $0\le l\le m-1$, such that $3^li\pmod n \equiv(\ldots, 0,0)\in C_i$ (or $3^li\pmod n \equiv(\ldots, 2,2)\in C_i$), so $3^li\pmod n <\delta_2$ (or $n-3^li\pmod n <\delta_2$, respectively). Hence, the absolute coset leader in $C_i$ is less than $\delta_2$.

Therefore $\delta_2$ is the second largest absolute coset leader.

This completes the proof.
\end{proof}

\begin{theorem} \label{th3} {\rm Let $q=3$,   $m$  a positive integer, and   $n=q^m-1$.

(1) If $m\equiv 0\pmod 4$ and $m\ge 4$, then $\delta_3=\frac{3^m-1}{5}$ is the  {\bf third largest absolute coset leader},  $C_{\delta_3}=\{\delta_3,2\delta_3,n-3\delta_3,n-2\delta_3\}$, and  $|C_{\delta_3}|=4$.

(2) If $m\equiv 2\pmod 4$ and $m\ge 6$, then  $\delta_3=\frac{3^{m-6}-1}{5}+3^{m-6}+2\cdot3^{m-5}+2\cdot3^{m-3}+3^{m-2}$  is the  {\bf third largest absolute coset leader}, $C_{\delta_3}\ne C_{n-\delta_3}$,  and  $|C_{\delta_3}|=|C_{n-\delta_3}|=m$.}
\end{theorem}
\begin{proof}

(1) If $m\equiv 0\pmod 4$ and the $3$-adic expansion of $\delta_3$ is as follows:
   $$\delta_3=\frac{3^m-1}5=(1+2\cdot3+3^2)(1+3^4+\ldots+3^{\frac{m-4}{4}})=(\underbrace{1,2,1,0,\ldots,1,2,1,0}_{m/4~(1,2,1,0)'s}).$$

Firstly, it is  checked that  $C_{\delta_3}=\{\delta_3, 2\delta_3, n-2\delta_3, n-3\delta_3\}=C_{n-\delta_3}$, $|C_{\delta_3}|=4$ and  $\delta_3$ is the absolute  coset leader of the $q$-cyclotomic cosets $C_{\delta_3}$.

Secondly,  we prove that $\delta_3$ is the third largest absolute coset leader.

For $1\le i\le n-1$, the $3$-adic expansion of $i$ is as follows:
$i=(i_0,i_1,\ldots, i_{m-1})$, which has at least two elements among $0,1,2$.

If the expansion of $i$ has a  consecutive form:  $(00)$ or $(22)$. Then there is an integer $l$, $0\le l\le m-1$, such that $3^li\pmod n \equiv(\ldots, 0,0)\in C_i$ (or $3^li\pmod n \equiv(\ldots, 2,2)\in C_i$), so $3^li\pmod n <\delta_3$ (or $n-3^li\pmod n <\delta_3$, respectively). Hence, the absolute coset leader in $C_i$ is less than $\delta_3$.

If the expansion of $i$ has a  consecutive form: $(11)$  and the expansion of  $i$ has  the form: $(110)$ or $(112)$. Then   there is an integer $l$, $0\le l\le m-1$, such that $3^li\pmod n \equiv(\ldots, 1,1,0)\in C_i$ (or $3^li\pmod n \equiv(\ldots, 1,1,2)\in C_i$), so $3^li\pmod n <\delta_3$ (or $n-3^li\pmod n <\delta_3$, respectively). Hence, the absolute coset leader in $C_i$ is less than $\delta_3$.

If  the expansion of $i$ has not any consecutive form:  $(00)$,  $(11)$, or $(22)$,
 and it has a form: $(010)$ or $(212)$. Then we can easily check that the absolute coset leader of $C_i$ is less than $\delta_3$.

If  the expansion of $i$ has not any form: $(00)$, $(11)$, $(22)$,  $(010)$ and $(212)$, we will prove that the absolute coset leader of $C_i$ is less than $\delta_3$.
By the assumptions,
the expansion of $i$ is equivalent to insert some 1's into the sequence $(2,0,\ldots,2,0)$ (or $(0,2,\ldots, 0,2)$).
Since $m\equiv 0\pmod 4$,   the number of 1's in the expansion of $i$ is an even integer {$k$}.

If $k=0$, then $i=\delta_2$ (or $n-\delta_2$).

If $k=\frac m2$, then $i=\delta_3$ or $3i\pmod n \equiv \delta_3$.

If $2\le k<\frac m2$, then  $i=(\ldots,2,0,2,1,0,\ldots)$ (or $i=(\ldots,0,2,0,1,2,\ldots)$). Hence  there is an integer $l$, $0\le l\le m-1$, such that $3^li\pmod n \equiv(\ldots,2,0,2,1,0)\in C_i$ (or $3^li\pmod n \equiv(\ldots,0,2,0,1,2)\in C_i$), so $3^li\pmod n <\delta_3$ (or $n-3^li\pmod n <\delta_3$, respectively). So the absolute coset leader of $C_i$ is smaller than $\delta_3$.

Therefore $\delta_3$ is the third largest absolute coset leader when $m\equiv 0\pmod 4$.

(2) If $m\equiv 2\pmod 4$ and the 3-adic expansion of $\delta_3$ is as follows:
\begin{eqnarray*}
 \delta_3&=&(1+2\cdot3+3^2)(1+3^4+\ldots+3^{\frac{m-10}{4}})+3^{m-6}+2\cdot3^{m-5}+2\cdot3^{m-3}+3^{m-2}\\
&=&(\underbrace{1,2,1,0,\ldots,1,2,1,0}_{(m-6)/4~(1,2,1,0)'s},1,2,0,2,1,0).
\end{eqnarray*}
In fact, the number of $1's$ in the expansion of $\delta_3$ is  $\frac{m-2}2$.

Firstly, we prove that $\delta_3$ is the absolute  coset leader of the $q$-cyclotomic cosets $C_{\delta_3}$ and $C_{n-\delta_3}$.
For $ 1 \leq l \leq m-1$, $3^l\delta_3 \pmod n$ are all different and $\delta_3$ is the smallest one in $C_{\delta_3}$. Hence $C_{\delta_3}$ has $m$ distinct elements, i.e. $|C_{\delta_3}|=m$, and  $\delta_3$ is the absolute coset leader in $C_{\delta_3}$.
Similarly, we can prove that $|C_{n-\delta_3}|=m$, $C_{\delta_3}\ne C_{n-\delta_3}$, and $C_{n-\delta_3}$ has also the absolute coset leader $\delta_3$.

Secondly,  we prove that $\delta_3$ is the third largest absolute coset leader.

 For $1\le i\le n-1$, there is a $3$-adic expansion:
$i=(i_0,i_1,\ldots, i_{m-1})$, which has at least two elements among $0,1,2$.

If the expansion of $i$ has a  consecutive form:  $(00)$ or $(22)$. Then there is an integer $l$, $0\le l\le m-1$, such that $3^li\pmod n \equiv(\ldots, 0,0)\in C_i$ (or $3^li\pmod n \equiv(\ldots, 2,2)\in C_i$), so $3^li\pmod n <\delta_3$ (or $n-3^li\pmod n <\delta_3$, respectively). Hence, the absolute coset leader in $C_i$ is less than $\delta_3$.

If the expansion of $i$ has a  consecutive form: $(11)$ and  the expansion of $i$ has  a form:  $(110)$ or $(112)$. Then  there is an integer $l$, $0\le l\le m-1$, such that $3^li\pmod n \equiv(\ldots, 1,1,0)\in C_i$ (or $3^li\pmod n \equiv(\ldots, 1,1,2)\in C_i$), so $3^li\pmod n <\delta_3$ (or $n-3^li\pmod n <\delta_3$, respectively). Hence, the absolute coset leader in $C_i$ is less than $\delta_3$.

If  the expansion of $i$ has not any  consecutive form:  $(00)$,   $(11)$,  or $(22)$, and it has a  form: $(010)$ or $(212)$. Then  we can easily check that the absolute coset leader of $C_i$ is less than $\delta_3$.

If  the expansion of $i$ has not any form: $(00)$, $(11)$, $(22)$,  $(010)$ and $(212)$, we will prove that the absolute coset leader of $C_i$ is less than $\delta_3$.
Similarly, by the assumptions, the expansion of $i$ is equivalent to insert some 1's into the sequence $(2,0,\ldots,2,0)$ (or $(0,2,\ldots, 0,2)$).
Since $m\equiv 2\pmod 4$,   the number of 1's in the expansion of $i$ is an even integer  $k$ with $0\le k\le \frac{m-2}2$.

If $k=0$, then $i=\delta_2$.

If $k=\frac{m-2}2$ and the expansion of $i$ has only one form: $(202)$ or $ (020)$, i.e. $i=(\ldots,1, 2, 1,0,\ldots,1,0,1,\underbrace{2,0, 2},1,0,\ldots )$ or $i=(\ldots,1, 2, 1,0,\ldots,1,2,1,\underbrace{0, 2,0},1,2\ldots )$. Then there is an integer $l$, $0\le l\le m-1$, such that $3^li\pmod n=\delta_3$ or $n-3^li\pmod n=\delta_3$.

If $k=\frac{m-2}2$ and the expansion of $i$ has not any  form: $(202)$ or $(020)$, i.e.,  $i=(\ldots, 2, 1,\underbrace{0, 2},1,0,\ldots )$ (or $i=(\ldots,0, 1, \underbrace{2, 0},1, 2,\ldots)$). Then there is a integer $l$, $0\le l\le m-1$, such that $3^li=(\ldots, 1, 0, 2, 1, 0)<\delta_3$ (or $n-3^ii=n-(\ldots, 1, 2, 0, 1, 2)<\delta_3$, respectively). Hence the absolute coset leader in $C_i$ is less than $\delta_3$.

If $2\le k<\frac{m-2}2$. We consider the following some cases.

(I) If  the expansion of $i$ has  one of the following six cases:
\begin{eqnarray*}i&=&(\ldots,2,1,\underbrace{0,2}_2,1,0,\ldots),~
i=(\ldots,0, 1,\underbrace{2, 0}_{2},1,2,\ldots), \\i&=&(\ldots,1, \underbrace{0, 2,\ldots0,2,0, 2}_{>3}, 1, \ldots), ~i=(\ldots,1, \underbrace{2,0,\ldots, 2,0,2,0}_{>3}, 1, \ldots),\\i&=&(\ldots,1, \underbrace{2,0,\ldots0,2,0, 2}_{>3}, 1, \ldots), ~i=(\ldots,1, \underbrace{0,2,\ldots, 0,2,0}_{>3}, 1, \ldots),\end{eqnarray*}
i.e. there are  two  or more than three elements  between two $1's$.  Then there is an integer $l$ such that $3^li<\delta_3$ or $n-3^li<\delta_3$.

(II) If the expansion of $i$ has a form:
$$i=(\ldots, 1, \underbrace {2,0, 2}, 1,\underbrace{0,2,0},1,\ldots),$$
where each $1$ inserts between $(2,0,2)$ and $(0,2,0)$.  Then $m=k+3k=4k$, which is contradictory.

(III) If the expansion of $i$ has a form:
$$i=(\ldots, 1, \underbrace {2,0, 2}, 1,0,1, 2, 1,\underbrace{0,2,0},1,\ldots),$$
where $0$, $2$, $(202)$, and $(020)$ appear between two $1's$. Let the number of $(202)$ and $(020)$ in the expansion of $i$ be $t$, then $t$ is odd.

In fact, if
   $(202)$ and $(020)$ are viewed as $2$ and $0$, respectively, i.e.  $$i'=(\ldots, 1, \underbrace {2}, 1,0,1, 2, 1,\underbrace{0},1,\ldots).$$
   Then by $m\equiv2\pmod 4$ and $k$ even, $m=2k+2t$ and $t$ is odd.

Without loss of generality, there are two adjacent   $(202)'s$ in the expansion of $i$, i.e.,
$$i=(\ldots,\underbrace{2,0, 2},1, 0, \ldots,0,1,\underbrace{2,0,2},1,0,\ldots).$$
Then there is an integer $l$ such that $3^li<\delta_3$.

Hence the absolute coset leader in $C_i$ is less than $\delta_3$.
Therefore $\delta_3$ is the third largest absolute coset leader for $m \equiv 2 \pmod 4$.

This completes the proof.
\end{proof}

\section{Parameters of some BCH codes}

In this section, we will first present three classes of ternary BCH codes,  determine their parameters and weight distributions. Secondly,  four classes of ternary LCD BCH codes are  proposed,   weight distributions of two of these codes are calculated and the others convert to the calculations of the Kloosterman sums.

We always assume that $n=3^m-1$, $\alpha$ is a primitive element of $\Bbb F_{3^m}$, and $C_i$ is the $3$-cyclotomic coset. We shall compute the weight distributions of BCH codes.

\subsection{Three classes of BCH Codes and their weight distributions}

\begin{theorem}\label{th4i}{\rm

Let $m$ be an odd integer, $\delta_1=\frac{3^m-1}{2}$, $\delta_2=\frac{3^{m-1}-1}{4}+3^{m-2}$, $Z=\bigcup_{-\delta_2 < s\leq \delta_1}C_s$, and $g(x)=\prod_{i\in Z}(x-\alpha^i)$. Then
$$\mathcal C_{(3,3^m-1,\delta_1+\delta_2+1,-\delta_2)}=\{c(a)=(\Tr_{3^m/3}(a\alpha^{\delta_2i}))_{i=0}^{n-1}:a\in \Bbb F_{3^m}\}$$ is a one-weight $[3^m-1,m,2\cdot 3^{m-1}]$ BCH code. }
\end{theorem}

\begin{proof}
  By Theorem \ref{th2}, $\delta_2$ is the second largest abstract coset leader  if $m$ is odd, the parity-check  polynomial of $\mathcal C_{(3,3^m-1,\delta_1+\delta_2,-\delta_2)}$ is  $h(x)=\prod_{i\in C_{n-\delta_2}}(x-\alpha^i)$, which is irreducible over $\Bbb F_q$ and  $\deg(h(x))=m$, so the dimension of the code is $m$.

For $a\in \Bbb F_{3^m}^*$, let $\omega$ be a 3-th primitive root of unit in the complex field.   Since $\frac{3^m+1}{4}\in C_{n-\delta_2}$, and $(\frac{3^m+1}{4},3^m-1)=1$, we have
\begin{eqnarray*}&&W_H(c(a))=n-\frac 13\sum_{y\in \Bbb F_3}\sum_{i=0}^{n-1}\omega^{y\Tr_{3^m/3}(a\alpha^{\delta_2i})}
\\&=&\frac {2n}{3}-\frac 1{3}\sum_{y\in \Bbb F_3^*}\sum_{i=0}^{n-1}\omega^{y\Tr_{3^m/3}(a\alpha^{\frac{3^m+1}{4}i})}
\\&=&\frac {2n}{3}-\frac 1{3}\sum_{y\in \Bbb F_3^*}\sum_{x\in \Bbb F_{3^m}^*}\omega^{y\Tr_{3^m/3}(ax)}=\frac {2n}{3}-\frac 2{3}(\sum_{x\in \Bbb F_{3^m}}\omega^{\Tr_{3^m/3}(ax)}-1)=2\cdot3^{m-1}.
\end{eqnarray*}

This completes the proof.
\end{proof}

\begin{example}{\rm
Let $p=3$, $m=5$, and $n=p^m-1=242$. Then the BCH code in Theorem \ref{th4i} has weight enumerator $1+242z^{162},$
which is confirmed by Magma.}
\end{example}

\begin{theorem}\label{th4ii}{\rm
Let $m \geq 6$ be an even integer with $m\equiv2\pmod  4$, $\delta_1=\frac{3^m-1}{2}$, $\delta_3=\frac{3^{m-6}-1}{5}+3^{m-6}+2\cdot3^{m-5}+2\cdot3^{m-3}+3^{m-2}$, $Z=(\bigcup_{-\delta_3< s\leq \delta_1}C_s)$,  and $g(x)=\prod_{i\in Z}(x-\alpha^i)$. Then
\begin{equation*}
    \mathcal C_{(3,3^m-1,\delta_1+\delta_3+1,-\delta_3)}=\{c(a)=(\Tr_{3^m/3}(a\alpha^{\delta_3i}))_{i=0}^{n-1}:a\in \Bbb F_{3^m}\}
\end{equation*}
is a BCH code with  parameters $[3^m-1,m,\frac{2}3\cdot(3^m-3^{\frac m2})]$ and
the weight distribution   in Table 1.
\begin{table}[!htbp]
\text{Table 1}\\
\begin{tabular}{c|c}
  \hline
  Weight & Frequency \\ \hline
  0 & 1 \\ \hline
  $\frac{2}3\cdot(3^m-3^{\frac m2})$& $\frac{3^m-1}{2} $ \\ \hline
   $\frac{2}3\cdot(3^m+3^{\frac m2})$& $\frac{3^m-1}{2} $ \\ \hline
\end{tabular}
\end{table}
}
\end{theorem}

\begin{proof}
By Theorem \ref{th3}, $\delta_3$ is the third largest abstract coset leader, the parity-check polynomial of $\mathcal C_{(3,3^m-1,\delta_1+\delta_3+1,-\delta_3)}$ is  $h(x)=\prod_{i\in C_{n-\delta_3}}(x-\alpha^i)$, so the dimension of the code is $m$.

For $a\in \Bbb F_{3^m}^*$, let $\omega$ be a 3-th primitive root of unit in the complex field. By $m\equiv 2\pmod 4$, $\alpha^{\frac{3^m-1}{3-1}}\in (\Bbb F_{3^m}^*)^2$ and $\Bbb F_3^*\subset  (\Bbb F_{3^m}^*)^2$.   Since $\frac{3^m-19}{5}\in C_{n-\delta_3}$ and $\gcd(\frac{3^m-19}{5},3^m-1)=2$, for $0\ne a\in \Bbb F_{3^m}$,  \begin{eqnarray*}&&W_H(c(a))=n-\frac 13\sum_{y\in \Bbb F_3}\sum_{i=0}^{n-1}\omega^{y\Tr_{3^m/3}(a\alpha^{\delta_3i})}=\frac {2n}{3}-\frac 1{3}\sum_{y\in \Bbb F_3^*}\sum_{i=0}^{n-1}\omega^{y\Tr_{3^m/3}(a\alpha^{\frac{3^m-19}{5}i})}
\\&=&\frac {2n}{3}-\frac 1{3}\sum_{y\in \Bbb F_3^*}\sum_{x\in \Bbb F_{3^m}^*}\chi(yax^2)
=\frac {2n}{3}-\frac 2{3}(\sum_{x\in \Bbb F_{3^m}}\chi(yax^2)-1)
\\&=&\frac {2\cdot 3^m}{3}-\frac 2{3}\eta(a)G(\eta)
 \\&=&\left\{\begin{array}{ll} \frac{2}3\cdot(3^m-3^{\frac m2}), &\mbox{ if }a \mbox{ is a square },\\ \frac{2}3\cdot(3^m+3^{\frac m2}), &\mbox{ if }a\mbox{ is not a square },\end{array}\right.
\end{eqnarray*}
where  $\eta$ is the multiplicative character of order $2$ over $\Bbb F_{3^m}$.
Hence the frequency of the weights is easy to obtain and  this completes the proof.
\end{proof}

\begin{example}{\rm
Let $p=3$, $m=6$, and $n=p^m-1=728$. Then the BCH code in Theorem \ref{th4ii} has weight enumerator
$1+364z^{468}+364z^{504},$
which is confirmed by Magma.
}
\end{example}
\begin{theorem}\label{th4iii} {\rm
Let $m$ be an integer with $m\equiv2\pmod 4$, $\delta_1=\frac{3^m-1}{2}$ , $\delta_3=\frac{3^{m-6}-1}{5}+3^{m-6}+2\cdot3^{m-5}+2\cdot3^{m-3}+3^{m-2}$ , $Z=(\bigcup_{-\delta_3<s<\delta_1}C_s)$, and $g(x)=\prod_{i\in Z}(x-\alpha^i)$. Then
$$\mathcal C_{(3,3^m-1,\delta_1+\delta_3,-\delta_3)}=\{c(a,b)=\left(a(-1)^i+\Tr_{3^m/3}(b\alpha^{\delta_3i})\right)_{i=0}^{n-1}: a\in \Bbb F_3 ,b\in \Bbb F_{3^m}\}$$ is a ternary BCH code with  parameters $[3^m-1,m+1,\frac{2}3 \cdot(3^m-3^{\frac{m}2})]$ and the  weight distribution   in Table 2.
\begin{table}[!htbp]
 \text{Table 2}\\
\begin{tabular}{c|c}
  \hline
  Weight  & Frequency \\ \hline
   0 & 1 \\ \hline
  $\frac{2}3 \cdot(3^m-3^{\frac{m}2})$ & $\frac{3^m-1}2$ \\ \hline
 $\frac{2}3 \cdot(3^m+3^{\frac{m}2})$ & $\frac{3^m-1}2 $\\ \hline
 $ \frac{1}3(2\cdot3^m+3^\frac{m}2)-1 $& $3^m-1$ \\ \hline
  $\frac{1}3(2\cdot3^m-3^\frac{m}2)-1$ & $3^m-1$ \\ \hline
  $3^m-1$ & 2 \\ \hline
\end{tabular}
\label{t1}
\end{table}
}
\end{theorem}
\begin{proof}
By Theorem \ref{th3}, $\delta_3$ is the third largest abstract coset leader, the parity-check polynomial of $\mathcal C$ is  $h(x)=(x+1)\prod_{i\in C_{n-\delta_3}}(x-\alpha^i)$, so the dimension of the code is $m+1$.

 Let $\omega$ be a 3-th primitive root of unit in the complex field. By $m\equiv 2\pmod 4$,  $\Bbb F_3^*\subset  (\Bbb F_{3^m}^*)^2$;  by $-\frac{3^m-19}{5}\in C_{n-\delta_3}$,  $(\frac{3^m-19}{5},3^m-1)=2$.  For $a\in \Bbb F_3$ and  $b\in \Bbb F_{3^m}$,
\begin{eqnarray*}W_H(c(a,b))&=&n-\frac 13\sum_{y\in \Bbb F_3}\sum_{i=0}^{n-1}\omega^{y[a(-1)^i+\Tr_{3^m/3}(b\alpha^{\delta_3i})]}
\\&=&\frac {2n}{3}-\frac 1{3}\sum_{y\in \Bbb F_3^*}\sum_{i=0}^{n-1}\omega^{y[a(-1)^i+\Tr_{3^m/3}(b\alpha^{\delta_3i})]}
\\&=&\frac {2n}{3}-\frac 1{3}\sum_{y\in \Bbb F_3^*}\omega^{ya}\sum_{x\in\Bbb F_{3^m}^*}\omega^{\Tr_{3^m/3}(bx^{2})}.
\end{eqnarray*}

Suppose that  $a=0$ and $b=0$. Then $W_H(c(a,b))=0$.

Suppose that  $a\neq 0$ and $b= 0$. Then
\begin{eqnarray*}W_H(c(a,b))&=&\frac {2n}{3}-\frac n{3}\sum_{y\in \Bbb F_3^*}\omega^{ya}=n.
\end{eqnarray*}

Suppose that  $a=0$ and $b\neq 0$.  Then
\begin{eqnarray*}
    &&W_H(c(a,b))=\frac {2n}{3}-\frac 2{3}(\sum_{x\in \Bbb F_{3^m} }\omega^{\Tr_{3^m/3}(bx^2)}-1)
        =2\cdot 3^{m-1}-\frac 2{3}\eta(b)G(\eta)
\\ &=&\left\{\begin{array}{ll}\frac{2}3\cdot(3^m-3^{\frac m2}), &\mbox{ if }b \mbox{ is a square},\\ \frac{2}3\cdot(3^m+3^{\frac m2}), &\mbox{ if }b  \mbox{ is not a square}.\end{array}\right.
\end{eqnarray*}

Suppose that  $a\neq 0$ and $b\neq 0$.  Then
\begin{eqnarray*}
    &&W_H(c(a,b))=\frac {2n}{3}-\frac 1{3}\sum_{y\in \Bbb F_3^*}\omega^{ya}\sum_{x\in\Bbb F_{3^m}^*}\omega^{\Tr_{3^m/3}(bx^{2})}
    \\&=&\frac {2n}{3}-\frac 1{3}(-1)\sum_{x\in \Bbb F_{3^m} }(\chi(\Tr_{3^m/3}(bx^2))-1)
           \\&=&
         \frac {2n}{3}-\frac 1{3}+\frac 1{3}\eta(b)G(\eta)
\\ &=&\left\{\begin{array}{ll}\frac{1}3(2\cdot 3^m+3^{\frac m2})-1, &\mbox{ if }b \mbox{ is a square },\\ \frac{1}3(2\cdot 3^m-3^{\frac m2})-1, &\mbox{ if }b  \mbox{ is not a square. }\end{array}\right.
\end{eqnarray*}

Note that  it is easy to obtain their frequencies and this completes the proof.
\end{proof}

\begin{example}{\rm
Let $p=3$, $m=6$, and $n=p^m-1=728$. Then the BCH code in Theorem \ref{th4iii} has weight enumerator
$$1+364z^{468}+728z^{476}+728z^{494}+364z^{504}+2z^{728},$$
which is confirmed by Magma.}
\end{example}

\subsection{Ternary LCD BCH Codes}
$\\$
Let $n=3^m-1$ and $\alpha$ a primitive element of $\Bbb F_{3^m}$. Define a ternary LCD BCH code $\mathcal C_{(3,n,-t,2t)}=\langle g(x)\rangle$, where $t$ is a positive integer, $Z=\bigcup_{|i|< t}C_i$ is a defining set, and
$g(x)=\prod_{i\in Z}(x-\alpha^i)$.
Now we shall choose some $t$ to compute the weight distributions of the  ternary LCD BCH cyclic codes.

\begin{theorem} {\rm Let $m$ be an integer and $\delta_1=\frac{3^{m}-1}{2}$.
Then
$$\mathcal C_{(3,3^m-1, 2\delta_1,-\delta_1)}=\{c(a)=(\Tr_{3^m/3}(ax))_{x\in \Bbb F_{3^m}^*}: a\in \Bbb F_{3^m}\}$$ is a ternary LCD BCH cyclic code with  parameters $[3^m-1,1,3^m-1]$ and its designed distance $3^m-1$.}
\end{theorem}
\begin{proof}   By Theorem \ref{th1}, $\delta_1$ is the largest abstract coset leader, the parity-check polynomial of $\mathcal C_{(3,3^m-1,2\delta_1,-\delta_1)}$ is  $h(x)=\frac{x^n-1}{g(x)}=x+1$, where $h(x)$ is  irreducible over $\Bbb F_3$, if $\alpha$ is an $n$th root of unit in $\Bbb F_{3^m}$, $h(\alpha^{\delta_1})=0$, $\deg(h(x))=1$,  and $h(x)$ is a self-reciprocal polynomial.

Let $\beta=\alpha^{\delta_1}$. Then
$$\mathcal C_{(3,3^m-1,2\delta_1,-\delta_1)}=\{c(a)=(a\beta^i)_{i=0}^{n-1}: a\in \Bbb F_{3}\}.$$

So, it has parameters $[3^m-1,1,3^m-1 ]$ and it has one all zeros codeword and two codewords with weight $3^m-1$.
\end{proof}

\begin{theorem} \label{th4iiii} {\rm Let $m$ be an odd integer, $\delta_1=\frac{3^m-1}{2}$, $\delta_2=\frac{3^{m-1}-1}{4}+3^{m-2}$, $Z=(\bigcup_{|s|<\delta_2}C_s)\bigcup C_{\delta_1}$, and $g(x)=\prod_{i\in Z}(x-\alpha^i)$. Then
$$\mathcal C_{(3,3^m-1,2\delta_2,-\delta_2)}=\{c(a,b)=(\Tr_{3^m/3}(ax+bx^{-1}))_{x\in \Bbb F_{3^m}^*}: a,b\in \Bbb F_{3^m}\}$$ is a ternary LCD  cyclic code with  parameters $[3^m-1,2m,\ge 2\delta_2]$ and its designed distance $2\delta_2$.}
\end{theorem}

\begin{proof}   By Theorem \ref{th2}, $\delta_2$ is the second largest absolute coset leader, the parity-check polynomial of $\mathcal C_{(3,3^m-1,2\delta_2,-\delta_2)}$ is  $h(x)=\frac{x^n-1}{g(x)}=f(x)\widehat f(x)$, where $f(x)$ is  irreducible over $\Bbb F_3$, $f(\alpha^{\delta_2})=0$, $\deg(f(x))=m$,  and $\widehat f(x)$ is a reciprocal polynomial of $f(x)$.

Let $\beta=\alpha^{\delta_2}$. Then
by Delsarte's Theorem \cite{D1},
$$\mathcal C_{(3,3^m-1,2\delta_2,-\delta_2)}=\{c(a,b)=(\Tr_{3^m/3}(a\beta^i+b(\beta^{-1})^i))_{i=0}^{n-1}: a,b\in \Bbb F_{3^m}\}.$$

On the other hand, by $m$ odd, $-\frac{3^m+1}{4}\in C_{\delta_2}$ and  $\gcd(4,3^m+1)=1$, we get that $\gcd(\delta_2, 3^m-1)=1$ and $\beta$ is a primitive element of $\Bbb F_{3^m}$.  Hence
$$\mathcal C_{(3,3^m-1,2\delta_2,-\delta_2)}=\{c(a,b)=(\Tr_{3^m/3}(ax+bx^{-1}))_{x\in \Bbb F^*_{3^m}}: a,b\in \Bbb F_{3^m}\}.$$
By Theorem \ref{th1} and BCH bound, it has parameters $[3^m-1,2m, \ge 2\delta_2]$.
\end{proof}

Let $a,b\in \Bbb F_{3^m}$,
 the Kloosterman sum $K_m(a, b)$ is defined over $\Bbb F_{3^m}$ as follows:
 $$K_m(a,b)=\sum_{x\in \Bbb F^*_{3^m}}\chi(ax+bx^{-1}),$$

 where $\chi$ is the canonical additive character of  $\Bbb F_{3^m}$.

\begin{corollary}{\rm Let $m$ be an odd integer. Then for $a,b\in \Bbb F_{3^m}$ and $(a,b)\ne(0,0)$,
$$K_m(a,b)\le \frac{3^m+2\cdot 3^{m-1}-1}{4}.$$}
\end{corollary}
\begin{proof} For $a,b\in\Bbb F_{3^m}$ and $(a,b)\ne (0,0)$,
by Theorem \ref{th4iiii},
\begin{eqnarray*}W_H(c(a,b))&=&n-|\{x\in \Bbb F_{3^m}^*:\Tr_{3^m/3}(ax+bx^{-1})=0 \}|\\
&=&n-\frac 13 \sum_{y\in \Bbb F_3}\sum_{ x\in \Bbb F_{3^m}^*}\chi(y(ax+bx^{-1}))\\
&=&\frac {2n} {3} -\frac {2} {3}K_m(a,b).
\end{eqnarray*}
Hence $\frac {2n} {3} -\frac {2} {3}K_m(a,b)\ge 2(\frac{3^{m-1}-1}4+3^{m-2})$ and $K_m(a,b)\le \frac{3^m+2\cdot 3^{m-1}-1}{4}$.
\end{proof}



\begin{remark} {\rm Numerical examples by Magma show that the bound here is not tight in general.}
\end{remark}

\begin{theorem}\label{th45} {\rm Let $m$ be an even integer, $\delta_1=\frac{3^{m}-1}{2}$,  $\delta_2=\frac{3^{m}-1}{4}$,  $Z=\bigcup_{|s|<\delta_2}C_s$, and $g(x)=\prod_{i\in Z}(x-\alpha^i)$.
Then   $$\mathcal C_{(3,3^m-1,2\delta_2,-\delta_2)}=\{c(a,b)=(a\alpha^{\delta_1i}+\Tr_{3^2/3}(b\alpha^{\delta_2i}))_{i=0}^{n-1}: a\in \Bbb F_{3}, b\in \Bbb F_{3^2}\}$$ is a ternary LCD BCH code with  parameters $[3^m-1,3, \frac12\cdot(3^m-1)]$ and the   weight distribution  in Table 3.
\begin{table}[!htbp]
\text{Table 3}\\
\begin{tabular}{c|c}
  \hline
  Weight  & Frequency \\ \hline
  0 & 1 \\ \hline
  $\frac12\cdot(3^m-1)$ & 12 \\ \hline
   $\frac34\cdot(3^m-1)$& 8 \\ \hline
  $3^m-1$& 6 \\ \hline
\end{tabular}
\label{t1}
\end{table}
}
\end{theorem}

\begin{proof}  By Theorem \ref{th2}, $\delta_2$ is the second largest abstract coset leader and the parity-check polynomial of $\mathcal C_{(3,3^m-1,2\delta_2,-\delta_2)}$ is  $h(x)=\frac{x^n-1}{g(x)}=f_1(x) f_2(x)$, where $f_1(x)$ and $f_2(x)$ are  irreducible over $\Bbb F_3$, $f_1(x)=x+1$,     $f_1(\delta_1)=0$, $f_2(x)=x^2+1$,  and $f_2(\delta_2)=0$.

Let $\zeta_4=\alpha^{\delta_2}\in \Bbb F_{3^2}$ be a $4$-th primitive root of unit and $\alpha^{\delta_1}=-1$.
Then
by Delsarte's Theorem \cite{D1},
$$\mathcal C_{(3,3^m-1,-\delta_2,2\delta_2)}=\{c(a)=(a(-1)^i)+\Tr_{3^2/3}(b\zeta_4^i))_{i=0}^{n-1}: a\in \Bbb F_{3}, b\in\Bbb F_{3^2}\}.$$

Let $\omega$ be a 3-th primitive root of unit in the complex field. By $m\equiv 0\pmod {4}$, $8|3^m-1$ and $\Bbb F_3^{*}\subset (\Bbb F_{3^2}^*)^2$. Denote  $Z(c(a,b))=|i\in \{0,1,\ldots,n-1\}:a(-1)^i+\Tr_{3^2/3}(b\zeta_4^i)=0|.$ Then
\begin{eqnarray*}
&&W_H(c(a,b))=n-Z(c(a,b))\\
&=&
n-\frac{1}{3}\sum_{y\in \Bbb F_3}\sum_{i=0}^{n-1}\omega^{y( a(-1)^i)+\Tr_{3^2/3}(b\zeta_4^i))}\\
&=&\frac{2n}3-\frac{n}{12}\sum_{y\in \Bbb F_3^*}\omega^{ay}\sum_{i=0}^{3}\omega^{\Tr_{3^2/3}(by(-\zeta_4)^i))}\\
&=&\frac{2n}{3}-\frac{n}{24}\sum_{y\in \Bbb F_3^*}\omega^{ay}\sum_{x\in \Bbb F_{3^2}^*}\omega^{\Tr_{3^2/3}(bx^2)}.
\end{eqnarray*}
Note that
 $(\Bbb F_{3^2}^*)^2=\langle \zeta_4\rangle$ and $\Bbb F_3^*\subset (\Bbb F_{3^2}^*)^2$.

Suppose that  $a=0$ and $b=0$. Then $W_H(c(a,b))=0$.

Suppose that  $a=0$ and $b\neq 0$.  Then by Lemma 2.3,
\begin{eqnarray*} W_H(c(a,b))&=&\frac{2n}{3}-\frac{n}{12}(\sum_{x\in\Bbb F_9}\omega^{y(\Tr_{3^2/3}(bx^2))}-1)=\frac {2n}3-\frac{n}{12}(\eta'(b)G(\eta')-1)\\
&=& \frac{3n}4-\frac{n}{12}\eta'(b)G(\eta')\\
&=&\left\{\begin{array}{ll}\frac{n}2, &\mbox{ if }b \mbox{ is a square, }\\ n, &\mbox{ if }b  \mbox{ is not a square. }\end{array}\right.\end{eqnarray*}
where $\eta'$ is a multiplicative character of order $2$ in $\Bbb F_9$.

Suppose that $a\ne 0$ and $b=0$. Then
\begin{eqnarray*} W_H(c(a,b))=\frac{2n}{3}-\frac{n}{24}\sum_{y\in\Bbb F_3^*}\omega^{ay}(3^2-1)=
n. \end{eqnarray*}

Suppose that $a\ne 0$ and $b\ne 0$. By Lemma 2.3,
\begin{eqnarray*} W_H(c(a,b))&=&\frac{2n}{3}-\frac{n}{24}(-1)(\sum_{x\in\Bbb F_9}\omega^{y(\Tr_{3^2/3}(bx^2))}-1)
\\ &=&\frac {2n}3+\frac{n}{24}(\eta'(b)G(\eta')-1)= \frac{3n}4+\frac{n}{24}\eta'(b)G(\eta')\\
&=&\left\{\begin{array}{ll}\frac{3n}4, &\mbox{ if }b \mbox{ is a square,  }\\ \frac n2, &\mbox{ if }b  \mbox{ is not a square. }\end{array}\right.\end{eqnarray*}

Note that it is easy to obtain their frequencies and this completes the proof.
\end{proof}

\begin{example}{\rm
Let $p=3$, $m=4$, and $n=p^m-1=81$. Then the LCD BCH code in Theorem \ref{th45} has weight enumerator $1+12z^{40}+8z^{60}+6z^{80},$
which is confirmed by Magma.}
\end{example}


\begin{theorem} {\rm  Let $m\equiv 2 \pmod 4$, $\delta_1=\frac{3^m-1}{2}$, $\delta_2=\frac{3^m-1}{4}$ $\delta_3=\frac{3^{m-6}-1}{5}+3^{m-6}+2\cdot3^{m-5}+2\cdot3^{m-3}+3^{m-2}$, $Z=(\bigcup_{|s|<\delta_3}C_s)\bigcup C_{\delta_1}\bigcup C_{\delta_2}$,
$g(x)=\prod_{i\in Z}(x-\alpha^i)$.
Then   $$\mathcal C_{(3,3^m-1,2\delta_3,-\delta_3)}
=\{c(a,b)=(\Tr_{3^m/3}(ax^2+bx^{-2}))_{x\in \Bbb F_{3^m}^*}: a,b\in \Bbb F_{3^m}\}$$ is a ternary LCD BCH code with  parameters $[3^m-1,2m, \ge 2\delta_3]$ and its designed distance $2\delta_3$. }
\end{theorem}
\begin{proof}   By Theorem \ref{th3}, $\delta_3$ is the third largest abstract coset leader, the parity-check polynomial of $\mathcal C_{(3,3^m-1,-\delta_3,2\delta_3)}$ is  $h(x)=\frac{x^n-1}{g(x)}=f(x)\widehat f(x)$, where $f(x)$ is  irreducible over $\Bbb F_3$, $f(\alpha^{\delta_3})=0$, $\deg(f(x))=m$,  and $\widehat f(x)$ is a reciprocal polynomial of $f(x)$.

Let $\beta=\alpha^{\delta_3}$. Then
by Delsarte's Theorem \cite{D1},
$$\mathcal C_{(3,3^m-1,2\delta_3,-\delta_3)}=\{c(a,b)=(\Tr_{3^m/3}(a\beta^i+b(\beta^{-1})^i))_{i=0}^{n-1}: a,b\in \Bbb F_{3^m}\}.$$

On the other hand, by $m\equiv 2 \pmod 4$ , $\frac{3^m-19}{5}\in C_{\delta_3}$ ,   $\gcd(5,3^m-1)=1$, and $\gcd(3^m-19, 3^m-1)=\gcd(18, 3^m-1)=2$, we get that $\gcd(\delta_3, 3^m-1)=2$ and $\beta$ is a semi-primitive element of $\Bbb F_{3^m}$.  Hence
$$\mathcal C_{(3,3^m-1,2\delta_3,-\delta_3)}=\{c(a,b)=(\Tr_{3^m/3}(ax^2+bx^{-2}))_{x\in \Bbb F^*_{3^m}}: a,b\in \Bbb F_{3^m}\}.$$
By Theorem \ref{th1} and BCH bound, the code has parameters $[3^m-1,2m, \ge 2\delta_3]$.
\end{proof}

\section{Concluding remarks}

In this paper, several classes of ternary primitive BCH codes and  LCD BCH codes were studied according to the first, second and third largest absolute coset leaders.  The weight distributions of these codes were given except two of them, whose weight distributions rely on the  calculation of  Kloosterman sums.

\section*{Acknowledgments}
The authors are very grateful to the reviewers and the Editor for their valuable  suggestions that  improved the quality of this paper.

\medskip

Received March 2021; 1st revision ; final revision .
\medskip

{\it E-mail address:} xinmeihuang@hotmail.com\\
\indent{\it E-mail address:} yueqin@nuaa.edu.cn\\
\indent{\it E-mail address:} yanshengwu@njupt.edu.cn\\
\indent{\it E-mail address:} xpshi@njfu.edu.cn

\end{document}